\documentclass[a4paper,11pt,twoside]{article}

\newcommand{\norm}{\alpha}
\newcommand{\CI}{\kappa}

\newcommand{\PTF}{\mathrm{PTF}}

\newcommand{\ph}{\mathrm{ph}}

\newcommand{\cc}{\mathrm{c}}
\newcommand{\C}{\mathrm{C}}
\newcommand{\G}{\mathrm{G}}
\newcommand{\LT}{\mathrm{LT}}
\newcommand{\romL}{\mathrm{L}}

\usepackage{exscale}%
\usepackage{amsmath}%
\usepackage{amsfonts}%
\usepackage{times}
\usepackage{amsthm}
\usepackage{enumerate}
\usepackage{epsfig}
\newcommand{\eps}{{\varepsilon}}        %%%%%%%%%%%%%%%%%%%%%%%%%%%%%%%%%%%
\newcommand{\la}{\langle} \newcommand{\ra}{\rangle}
 
\newcommand{\cD}{\mathcal{D}}
\newcommand{\cE}{\mathcal{E}}
\newcommand{\cF}{\mathcal{F}}

\newcommand{\cH}{\mathcal{H}}

\newcommand{\RR}{\mathbb{R}}            %%%%%%%%%%%%%%%%%%%%%%%%%%%
\renewcommand{\i}{\mathrm{i}}

\newcommand{\Tr}{\mathop{\mathrm{Tr}}}

\renewcommand{\thesection}%            %%%%%%%%%%%%%%%%%%%%%%%%%%%%%%%%%%%%
{\arabic{section}}                     % To begin a section, write \secct %
\renewcommand{\theequation}%           % instead of \section. The section %
{\thesection.\arabic{equation}}        % number will then be displayed    %
\newcommand{\secct}[1]{\section{#1}%   % as a roman number, e.g., IV. for %
\setcounter{equation}{0}}              % the fourth section               %
\newtheorem{theorem}{Theorem}[section]         %%%%%%%%%%%%%%%%%%%%%%%%%%%%%%
\newtheorem{lemma}[theorem]{Lemma}             % These Theoremlike environm.%
\newtheorem{proposition}[theorem]{Proposition} %%%%%%%%%%%%%%%%%%%%%%%%%%%%%%

\theoremstyle{plain}
\begin{document}

\bibliographystyle{amsplain}

%%%%%%%%%%%%%%%%%%%%%%%%%%%%%%%%%%%%%%%%%%%%%%%%%%%%%%%%%%%%%%%%%%%%%%%%%%%%
\setcounter{page}{0} \thispagestyle{empty}
\title{Bounds on the Minimal Energy of Translation Invariant $N$-Polaron Systems}
\author{
Marcel Griesemer\\
Universit\"at Stuttgart, Fachbereich Mathematik\\
70550 Stuttgart, Germany\\
{\em email:} marcel@mathematik.uni-stuttgart.de\\
\\
Jacob Schach M{\o}ller\\
%\footnote{Supported by a Skou stipend from the Danish Research Agency.}\\
Aarhus University, Department of Mathematical Sciences\\
8000 {\AA}rhus C, Denmark\\
{\em email:} jacob@imf.au.dk }

\date{}

\maketitle
\begin{abstract}
For systems of $N$ charged fermions (e.g. electrons) interacting with
longitudinal optical quantized lattice vibrations of a polar
crystal we derive upper and lower bounds on the minimal
energy within the model of H.~Fr\"ohlich. The only parameters of
this model, after removing the
ultraviolet cutoff, are the constants $U>0$ and $\alpha>0$ measuring the electron-electron and the
electron-phonon coupling strengths. They are constrained by the
condition  $\sqrt{2}\alpha<U$, which follows from 
the dependence of $U$ and $\alpha$ on electrical properties of the
crystal. We show that the large $N$ asymptotic behavior of the minimal
energy $E_N$ changes at $\sqrt{2}\alpha=U$ and that
$\sqrt{2}\alpha\leq U$ is necessary for thermodynamic stability: for $\sqrt{2}\alpha > U$ the phonon-mediated electron-electron attraction
overcomes the Coulomb repulsion and $E_N$ behaves like $-N^{7/3}$.
\end{abstract}

%\noindent {\bf MSC:} xxxx, yyyy.

%\vspace{2mm}

%\noindent {\bf Keywords:}

\thispagestyle{empty}

\setcounter{page}{1}

%%%%%%%%%%%%%%%%%%%%%%%%%%%%%%%%%%%%%%%%%%%%%%%%%%%%%%%%
%%%%%%%%%%%%%%%%%%%%%%%%%%%%%%%%%%%%%%%%%%%%%%%%%%%%%%%%
\newpage
\secct{Introduction} \label{sec-intro}
%%%%%%%%%%%%%%%%%%%%%%%%%%%%%%%%%%%%%%%%%%%%%%%%%%%%%%%%
%%%%%%%%%%%%%%%%%%%%%%%%%%%%%%%%%%%%%%%%%%%%%%%%%%%%%%%%

We study a system of $N$ electrons in a polar (ionic) crystal, modelled
by a Hamiltonian derived by H.~Fr{\"o}hlich \cite{Froehlich1954}. 
The model takes into account the electron-electron Coulomb repulsion, and 
a linear interaction of the electrons with the longitudinal optical phonons.
The model is called the 'large polaron' model, since it assumes that a
polaron 
(dressed electron) extends over a region which is large compared to the ion-ion spacing. 
In particular the underlying discrete (and infinite) crystal is replaced by a continuum.
See \cite{Devreese1996,Feynman1972,LeeLowPines1953}.

As is well-known, linear electron-phonon couplings
induce an effective pair attraction between electrons. 
This attraction competes with the electron-electron repulsion and may cause a 
phase-transition as the electron-phonon interaction strength increases.
This mechanism is behind the production of Cooper pairs in the BCS model
of low temperature superconductivity, and in high-$T_\cc$ superconductivity
the role of many-polaron systems is being investigated \cite{AlexandrovMott1996,DevreeseTempere1998,Hartingeretal2004}.

 The Fr{\"o}hlich Hamiltonian depends on two non-negative dimensionless quantities, $U$ and $\alpha$.
The constant $U$ is the electron-electron repulsion strength, and $\alpha$ is 
the Fr{\"o}hlich electron-phonon coupling constant. 
Physically relevant models must satisfy the constraint, 
cf. \cite{BrosensKliminDevreese2005,VerbistPeetersDevreese1991},
$$
  \sqrt{2}\alpha < U.
$$

 In this paper we prove upper and lower bounds on the minimal energy $E_N$
of the $N$-electron Fr{\"o}hlich Hamiltonian for all $N$ and all non-negative 
values of $U$, $\alpha$. In the unphysical regime $\sqrt{2}\alpha\geq U$, our results imply that
$E_N \sim -N^{7/3}$. In the physical regime we find that $E_N \geq -C N^2$, thus establishing a 
sharp transition in the large $N$-asymptotics of $E_N$ at $\sqrt{2}\alpha = U$. This transition 
is due to the mediated attraction between electrons overcoming the repulsion at $\sqrt{2}\alpha=U$ 
in the limit of large $N$. In fact, the quantity $U-\sqrt{2}\alpha$ appears in our analysis 
as an effective Coulomb coupling strength. 
We also demonstrate that $E_N\leq -\alpha N$ and $E_{N+M}\leq E_N+E_M$
in the physical regime. We do not  
know whether or not $E_N$ is an extensive quantity, but if it is not
extensive, then this must be due to electron-phonon correlations, 
cf. Proposition~\ref{mainA}.

We pause this discussion to introduce the mathematical model.
 The Fr{\"o}hlich Hamiltonian describing $N$ electrons in a polar crystal reads
\begin{equation}\label{HN}
   \sum_{\ell=1}^N\big[-\tfrac12\Delta_{x_\ell}+\sqrt{\alpha}\Phi(x_\ell)\big] + H_\ph 
 +U V_\C,
\end{equation}
where the number operator 
$$
   H_\ph = \int_{\RR^3}a^*(k)a(k) dk,
$$
accounts for the kinetic energy of the phonons while 
the field operator
$$
  \Phi(x) = \int_{\RR^3}\frac1{c_0|k|}\big[e^{\i k\cdot x}a(k) + e^{-\i k\cdot x}a^*(k)\big]dk,
$$
is responsible for the electron-phonon interaction. Here $c_0 := 2^{3/4}\pi$.
Finally the electron-electron interaction is given by the sum of
two-body Coulomb potentials
\[
V_\C(x_1,\ldots,x_N) =  \sum_{1\leq i< j \leq N} \frac{1}{|x_i-x_j|}.
\]
 We work in units where the frequency of the longitudinal
optical phonons, $\omega_{\mathrm{LO}}$, Planck's constant
$\hbar$, and the electron band mass are equal to one.

Let $\cF$ denote the symmetric Fock space over $L^2(\RR^3)$. The Hamiltonian \eqref{HN} defines a symmetric quadratic 
form on $\cH = \wedge^N L^2(\RR^3)\otimes\cF$, but, a priori, it is not well defined as a self-adjoint operator. For that one
must first impose an ultraviolet cutoff on the electron-phonon interaction:
Let $\Lambda>0$, and define the cutoff Hamiltonian as
$$
H_{N,\Lambda} = \sum_{\ell=1}^N-\tfrac12\Delta_{x_\ell} + H_\ph +
\sqrt{\alpha}\sum_{\ell=1}^N \Phi_\Lambda(x_\ell) + U V_\C,
$$
where
$$
\Phi_\Lambda(x) = \int_{|k|\leq \Lambda}\frac1{c_0|k|}
\big[e^{\i k\cdot x}a(k) +e^{-\i k\cdot x}a^*(k)\big]dk.
$$
The operators $H_{N,\Lambda}$ are self-adjoint on $\cD(H_\ph)\cap
\cD(\sum_{\ell=1}^N\Delta_{x_\ell})$, by the Kato-Rellich theorem, and it is well known, cf.
\cite{Ammari2000,Cannon1971,Froehlich1974,GerlachLowen1991,Nelson1964}, that
$H_{N,\Lambda}$ converges, as $\Lambda\to\infty$, in the norm-resolvent
sense to a self-adjoint operator, which we denote by $H_N$. This implies that
\begin{equation}\label{removecut}
E_N = \lim_{\Lambda\to\infty} E_{N,\Lambda}
\end{equation}
if $E_{N,\Lambda} := \inf\sigma(H_{N,\Lambda})$ and $E_N := \inf\sigma(H_N)$.

The main goal of this paper is to investigate the large $N$ behavior
of the minimal energy $E_N$ as a function of $\alpha$ and $U$.
Our first result is an upper bound in the regime $\sqrt{2}\alpha>U$.

\begin{theorem}\label{main1} There is a constant $C$ such that for all $N$
  and for $\sqrt{2}\alpha \geq U \geq 0$
\begin{equation*}
E_N \ \leq\  (\sqrt{2}\alpha-U)^2 N^\frac73\big[E_{\PTF} + C N^{-\frac1{17}}\big].
\end{equation*}
Here $E_{\PTF}<0$ is given by \eqref{eq:PTFenergy} below.
\end{theorem}

Theorem~\ref{main1} is proved variationally by using Pekar's ansatz in
terms of a product state, which is known to give the correct ground state energy for $N=1,2$ in the large $\alpha$ limit
\cite{DonskerVaradhan1983,LiebThomas1997,MiyaoSpohn2007}. Taking
the expectation value in a state $f\otimes\eta\in \wedge^N
L^2(\RR^3)\otimes\cF$ and explicitly minimizing with
respect to $\eta$ we arrive at a 
Hartree-Fock type energy which is then estimated by a Thomas-Fermi energy.
This allows us to scale out all parameters and we are left with the bound in 
Theorem~\ref{main1}, where 
\begin{equation}\label{eq:PTFenergy}
E_{\PTF} = \inf_{\rho\geq 0, \int\rho(x)dx=1} \cE_{\PTF}(\rho),
\end{equation}
\begin{equation}\label{eq:PTF}
   \cE_\PTF(\rho):= \tfrac{3}{10}(6\pi^2)^{\frac23}
   \int_{\RR^3}\rho(x)^{\frac53}\,dx - 
   \tfrac12\int_{\RR^6}\frac{\rho(x)\rho(y)}{|x-y|}\,dx
   dy.
\end{equation}
We note that in the error term in Theorem~\ref{main1} the exponent
$1/17$ can be replaced by any number less than $2/33$ at the expense of a larger and divergent constant $C$.

To show that the variational upper bound from
Theorem~\ref{main1} has the right asymptotics in $N$ and $\alpha$, we
provide the following lower bound: 

\begin{theorem}\label{main2} There exists $C >0$ such that
for all $N$ and $\sqrt{2}\alpha \geq U \geq 0$,
\begin{equation}\label{Main1lower}
E_N\geq -C_\G(\sqrt{2}\alpha - U)^2N^{\frac73} - C \alpha^2 N^{\frac73-\frac19}.
\end{equation}
\end{theorem}

This lower bound is obtained, essentially, by completing the square
with respect to creation and annihilation operators in 
the expression $H_\ph+\sqrt{\alpha}\sum_{j=1}^n\Phi(x_j)$. The
computation brings out an effective Coulomb interaction with coupling strength $-\sqrt{2}\alpha$.
Unfortunately, it also yields an infinite self-energy, which must be dealt with before 
completing the square. For that we use a commutator argument from \cite{LiebThomas1997}, 
which is responsible for the error term in Theorem~\ref{main2}. 
The resulting effective Hamiltonian with an attractive Coulomb
potential is bounded below by the 'gravitational collapse' bound 
\begin{equation}\label{Gravitation}
\sum_{j=1}^N-\tfrac12\Delta_j - \sum_{1\leq j<\ell\leq N}\frac1{|x_j-x_\ell|}
\geq -C_\G N^{\frac73}
\end{equation}
due to L{\'e}vy-Leblond \cite[Theorem 2]{LevyLeblond1968}. 
Hence the presence of the constant $C_\G$ in Theorem~\ref{main2}.

We now turn to the physical regime $\sqrt{2}\alpha<U$. Here our lower
bound is a byproduct of our proof of Theorem~\ref{main2}, and we have no reason to believe it is optimal.
Together with Theorem~\ref{main1} it demonstrates, however, that the model undergoes a 
sharp transition at $\alpha = U/\sqrt{2}$.

\begin{theorem}\label{main3} For $0< \sqrt{2}\alpha < U$,
\[
E_N\geq -\big(\tfrac{16}{3\pi}\alpha^2 N^2 + 3\big)\frac{U}{U-\sqrt{2}\alpha}.
\]
\end{theorem}

Last but not least there are the following universal variational upper
bounds for $E_N$ and $E_{N+M}$.

\begin{theorem}\label{main4} For all $N,M$,  $\alpha$ and $U$ we have
\begin{align*}
  E_N &\leq -\alpha N,\\
  E_{N+M} &\leq E_N+E_M.
\end{align*}
\end{theorem}

The bound $E_1\leq -\alpha$ is well known from \cite{LeeLowPines1953,
  Feynman1955} and it agrees with the result of a formal
computation of $E_1$ by second order perturbation theory
\cite{Feynman1972}. Also, it is consistent with Haga's computation of
$E_1$ including $\alpha^2$-terms\footnote{There is a sign error in
  Feynman's quote of Haga's result.} \cite{Haga1954}. The bound
$E_N\leq -\alpha N$ follows from the estimates $E_1\leq -\alpha$ and $E_N\leq
NE_1$, the latter of which is a consequence of 
the second result of Theorem~\ref{main4}. We remark that $E_{N+M}\leq E_N+E_M$ holds quite generally for translation
invariant $N$-particle systems with interactions that go to zero with
increasing particle separation. In particular it holds for
fermions and for distinguishable particles alike. 
Numerically computed upper bounds on $E(N)/N$, for $N=2$ through $N=32$ can be found in the literature
\cite{BrosensKliminDevreese2008}, but in the case of fermions they are
not refined enough to be consistent with the bound $E_{N+M}\leq E_N+E_M$.
%This raises some doubt about the validity of the conclusions in
%\cite{BrosensKliminDevreese2008} concerning the stability of
%multipolaron states.

In this paper we have omitted spin, but the Fermi statistics is taken
into account. There are only few small modifications necessary for treating 
fermions with $q$ spin states, such as factor of $q^{-2/3}$ in front of the Thomas-Fermi kinetic energy, 
which alters the upper bound in Theorem~\ref{main1} by a factor of $q^{2/3}$. 

The many-polaron model has also been studied with a confining potential of the form
$\sum_{\ell=1}^N W(x_\ell)$, $W(0)=0$ and $W\geq 0$ included in the
Hamiltonian \cite{KliminFominBrosensDevreese2004}. We could include such a potential in our work as well, but, at least in the regime
$\sqrt{2}\alpha> U$ this would not affect the leading large $N$
behaviour of $E_N$.

%our results then remain unchanged, if one modifies $\cE_\PTF$
%by a term $N^{-4/3}\int_{\RR^3} W(N^{-1/3}x)\rho(x)dx$. This term
%however disappears in the large $N$ limit, which reflects the fact that the density lives on the 
%lengthscale $N^{-1/3}$ and thus do not see the confining potential.
%This would be different in the physical regime.\\

%\noindent
%\emph{Acknowledgement:} The authors thank S{\o}ren Fournais for
%support by the ERC Starting Independent Researcher Grant 202859.

%%%%%%%%%%%%%%%%%%%%%%%%%%%%%%%%%%%%%%%%%%%%%%%%%%%%%%%%
%%%%%%%%%%%%%%%%%%%%%%%%%%%%%%%%%%%%%%%%%%%%%%%%%%%%%%%%

\secct{Upper bounds on $E_N$}
%%%%%%%%%%%%%%%%%%%%%%%%%%%%%%%%%%%%%%%%%%%%%%%%%%%%%%%%
%%%%%%%%%%%%%%%%%%%%%%%%%%%%%%%%%%%%%%%%%%%%%%%%%%%%%%%%

In this section we prove Theorem~\ref{main1} and Theorem~\ref{main4}. Since $E_N=\lim_{\Lambda\to\infty}E_{N,\Lambda}$ we only need to deal with 
the self-adjoint operator $H_{N,\Lambda}$. Let $f\in \cD_N = \wedge^N L^2(\RR^3)\cap H^1(\RR^{3N})$ be normalized and recall that
the one-particle density matrix $\gamma$ and the density function $\rho$ associated with 
$f$ are defined by
\begin{align}\label{eq:opdm}
    \gamma(x,x') &:= N\int_{\RR^{3(N-1)}} f(x,x_2,\dots,x_N)\overline{f(x',x_2,\dots,x_N)}
dx_2\cdots dx_N,\\
    \rho(x)&:= \gamma(x,x) = N\int_{\RR^{3(N-1)}} |f(x,x_2,\dots,x_N)|^2 dx_2\cdots
    dx_N.\label{eq:pd}
\end{align}
In this paper the Fouriertransform $\hat{\rho}$ of the density function $\rho$, or
of any other function, is defined by:
$$
     \hat{\rho}(k) = \int_{\RR^3} e^{-\i k\cdot x}\rho(x)dx,
$$ 
that is, without a factor of $(2\pi)^{-3/2}$.

\begin{proposition}\label{EN-bound}
Suppose $\sqrt{2}\alpha\geq U$. Then for every one-particle density matrix $\gamma$ on $L^2(\RR^3)$
with $0\leq \gamma\leq 1$, $\Tr[\gamma]=N$,
$\Tr[-\Delta\gamma]<\infty$, and for $\rho(x):=\gamma(x,x)$,
\begin{eqnarray*}
 E_N &\leq & (\sqrt{2}\alpha-U)^2\left[\tfrac12\Tr[-\Delta\gamma] -
 \tfrac{1}{2}\int_{\RR^6}\frac{\rho(x)\rho(y)}{|x-y|}\,dx dy\right]\\ 
&& -U(\sqrt{2}\alpha-U)\tfrac{1}{2}\int_{\RR^6}\frac{|\gamma(x,y)|^2}{|x-y|}\,dxdy.
\end{eqnarray*}
\end{proposition}

\begin{proof} 
This proof is based on the estimate $E_{N,\Lambda}\leq\la f\otimes
\eta, H_{N,\Lambda} f\otimes\eta\ra$ for suitable normalized $f\in\cD_N$ and $\eta\in\cF$. 
We begin by observing that the expectation value of the interaction
operator in a state $f\otimes\eta$ may be represented in the following
two ways: if $f$ and $\eta$ are normalized, then
\begin{eqnarray}
   \lefteqn{ \Big\la f\otimes \eta,\sum_{\ell=1}^N\Phi_{\Lambda}(x_\ell) f\otimes \eta\Big\ra}
\nonumber\\
    &=& \int_{\RR^{3N}} |f(x_1,\ldots,x_N)|^2 \sum_{\ell=1}^N V_{\Lambda,\eta}(x_\ell)\, dx_1\ldots dx_N\label{int1}\\
    &=&\la\eta, \Phi_{\Lambda}(\rho) \eta\ra \label{int2}
\end{eqnarray}
where $V_{\Lambda,\eta}(x):= \la\eta,\Phi_{\Lambda}(x) \eta \ra$,  $\rho$ is the density associated with $f$, and 
\begin{eqnarray*}
    \Phi_{\Lambda}(\rho) &:=&  \int_{\RR^3} \rho(x)
    \la\eta,\Phi_{\Lambda}(x)\eta\ra\, dx\\
    &=& \int_{|k|\leq \Lambda}\frac{1}{c_0|k|}\big[\overline{\hat\rho(k)}a(k)+\hat\rho(k)a^{*}(k)\big]dk.
\end{eqnarray*}
Hence if we define 
$H_{N,\Lambda}^{\eta}:=\sum_{\ell=1}^N[-\frac12\Delta_\ell+\sqrt{\alpha}V_{\Lambda,\eta}(x_\ell)]
+U V_C$, then 
\begin{equation}\label{eta-ham}
    \la f\otimes \eta, H_{N,\Lambda} f\otimes\eta \ra = \la f,  H_{N,\Lambda}^{\eta}f \ra + \la\eta, H_\ph\eta \ra.
\end{equation}
The ground state energy of the $N$-body Hamiltonian
$H_{N,\Lambda}^{\eta}$ is bounded above by its ground state energy
in the Hartree-Fock approximation. By Lieb's variational principle,
\cite{Lieb1981a} and \cite[Corollary 1]{Bach1992}, this Hartree-Fock ground state energy is bounded above by
\begin{equation}\label{HF-functional}
  \cE^{N,\Lambda}_{HF}(\gamma,\eta):=\Tr\left[\big(-\tfrac12\Delta+\sqrt{\alpha}V_{\Lambda,\eta}\big)\gamma\right]
  +\frac{U}{2}\int_{\RR^6}\frac{\rho(x)\rho(y)-|\gamma(x,y)|^2}{|x-y|}\,dxdy
\end{equation}
for any one-particle density matrix $\gamma$ with
$\Tr[\gamma]=N$ and $\rho(x)=\gamma(x,x)$. Hence, in view of \eqref{eta-ham}, we conclude that
\begin{equation}\label{HF-bound}
    E_{N,\Lambda} \leq \cE^{N,\Lambda}_{HF}(\gamma,\eta) + \la\eta, H_\ph\eta \ra
\end{equation}
for all normalized $\eta\in\cF$. In order to minimize the right hand
side with respect to $\eta$, we use that \eqref{int1} equals \eqref{int2}.
It follows, by Lemma~\ref{squares}, that 
\begin{equation}\label{complete-square}
     \inf_{\eta\in\cF,\|\eta\|=1} \big[\sqrt{\alpha}\Tr(V_{\Lambda,\eta}\gamma)+
     \la\eta, H_\ph\eta\ra\big] = -\frac{\alpha}{c_0^2}\int_{|k|\leq \Lambda}\frac{|\hat{\rho}(k)|^2}{|k|^2}\,dk.
\end{equation}
By combining \eqref{HF-functional},
\eqref{HF-bound}, and \eqref{complete-square} and then letting
$\Lambda\to\infty$ we arrive at
\begin{equation}\label{old-result}
  E_{N} \leq \tfrac12\Tr[-\Delta\gamma] +(U-\sqrt{2}\alpha)
 \tfrac{1}{2}\int_{\RR^6}\frac{\rho(x)\rho(y)}{|x-y|}\,dx dy 
 -\frac{U}{2}\int_{\RR^6}\frac{|\gamma(x,y)|^2}{|x-y|}\,dxdy
 \end{equation}
for any one-particle density matrix $\gamma$ with
$\Tr(\gamma)=N$ and $\rho(x)=\gamma(x,x)$. Here \eqref{Plancherel} and \eqref{removecut} were used also.
In the case
$\sqrt{2}\alpha=U$ it is clear from \eqref{old-result} or from
\eqref{eta-ham} with $\eta$ being the vacuum vector, that
$E_N\leq 0$. In the case where $\beta:=\sqrt{2}\alpha-U>0$, we choose
the density matrix $\gamma$ on the form
$\gamma=U_{\beta}\widetilde{\gamma}U_{\beta}^{*}$ with $U_{\beta}$
defined by $(U_{\beta}\varphi)(x):=\beta^{3/2}\varphi(\beta x)$. The Proposition
then follows from $U_{\beta}^{*}\Delta U_{\beta}=\beta^{2}\Delta$ and from
$\gamma(x,y)=\beta^3\widetilde{\gamma}(\beta x,\beta y)$ by a simple
change of variables in the integrals of \eqref{old-result}.
\end{proof}

%-------------------------------------------------------------------------------------

The second ingredient for proving Theorem~\ref{main1} is the following lemma. 

\begin{lemma}\label{lm:costate}
Let $g\in H^2(\RR^3)$ with $\|g\|=1$. Then for every 
$\rho\in L^1(\RR^3)$ with $\rho\geq 0$ and $\int_{\RR^3}\rho(x)dx=N$
there exists a density matrix $\gamma$ such that $\gamma(x,x)=(\rho*|g|^2)(x)$ and
\begin{equation*}
\Tr[-\Delta\gamma] = \tfrac{3}{5}(6\pi^2)^{\frac23}
\int_{\RR^3}\rho(x)^{\frac53}dx + N\|\nabla g\|^2.
\end{equation*}
\end{lemma}
%\noindent For the proof of this lemma see \cite[Page 621]{Lieb1981b}.

\begin{proof}
For the reader's convenience, we recall the proof from \cite[Page
621]{Lieb1981b}. Let $M:\RR^6\to \RR$ be defined by $M(p,q)=1$ if
$|p|\leq (6\pi)^{2/3}\rho(q)^{1/3}$ and $M(p,q)=0$ otherwise. Then 
\begin{eqnarray}
    (2\pi)^{-3}\int_{\RR^6} M(p,q)dpdq &=& \int_{\RR^3}\rho(q)dq = N\nonumber\\
 (2\pi)^{-3}\int_{\RR^6} p^2 M(p,q)dpdq &=& \tfrac{3}{5}(6\pi)^{\frac23}\int_{\RR^3}\rho(q)^{\frac53}dq.\label{phase-space-kE}
\end{eqnarray}
We define $\gamma$ by 
$$
    \gamma =  (2\pi)^{-3}\int_{\RR^6} M(p,q)\Pi_{pq}dpdq
$$
where $\Pi_{pq}$ is the rank one projection given by
$$
   \Pi_{pq}\varphi = g_{pq}\int_{\RR^3}\overline{g_{pq}(x)}\varphi(x)dx,\qquad g_{pq}(x)=e^{\i px}g(x-q).   
$$
It follows that $\gamma(x,x)=\int_{\RR^3}|g(x-q)|^2\rho(q)\,dq$, and 
from
\begin{align*}
  \Tr [-\Delta\Pi_{pq}] = \|\nabla g_{pq}\|^2 = p^2+\|\nabla g\|^2
  +2p\cdot \la g,-\i \nabla g\ra,
\end{align*}
and \eqref{phase-space-kE} we find the asserted expression for $\Tr[-\Delta\gamma]$. 
\end{proof}

\medskip
Proposition~\ref{EN-bound} and Lemma~\ref{lm:costate} suggest the definition of a 
\emph{Polaron Thomas-Fermi functional} by
\begin{equation}\label{eq:PTF2}
   \cE_\PTF(\rho):= \tfrac{3}{10}(6\pi^2)^{\frac23}
   \int_{\RR^3}\rho(x)^{\frac53}\,dx -
   \tfrac{1}{2}\int_{\RR^6}\frac{\rho(x)\rho(y)}{|x-y|}\,dx dy,
\end{equation}
where $\rho\in L^1(\RR^3)\cap L^{5/3}(\RR^3)$ and $\rho\geq
0$. If $\rho_N(x):=N^2\rho(N^{1/3}x)$, then $\|\rho_N\|_1=N\|\rho\|_1$ and 
$$
    \cE_\PTF(\rho_N) = N^{\frac73}\cE_\PTF(\rho). 
$$ 
Hence it suffices to consider
densities $\rho$ with $\int\rho(x)dx=1$. Let
\begin{equation*}
  E_\PTF := \inf\Big\{\cE_\PTF(\rho)\Big| \rho\geq 0,\ 
   \int_{\RR^3}\rho(x)dx=1\Big\}
\end{equation*}
which is finite by Lemma~\ref{lm:rho}.

%-------------------------------------------------------------------------------------

\begin{lemma}\label{lm:basicTF}
$E_\PTF <0$.
\end{lemma}

\begin{proof}
Given $\rho\in L^1(\RR^3)\cap L^{5/3}(\RR^3)$ with $\rho\geq 0$ and
$\int\rho\,dx=1$, let $\rho_R(x)=R^{-3}\rho(R^{-1}x)$. Then
$\int_{\RR^3}\rho_{R}(x)dx=1$ for all $R>0$ and 
$$
    \cE_\PTF(\rho_R) =
    R^{-2}\tfrac{3}{10}(6\pi^2)^{\frac23}\int_{\RR^3}\rho(x)^{\frac53}dx-R^{-1}\tfrac{1}{2}
    \int_{\RR^6}\frac{\rho(x)\rho(y)}{|x-y|}\,dx dy.
$$
This is negative for $R$ large enough.
\end{proof}

%In the case of $q$ spin states the pre-factor
%$(6\pi^2)^{\frac23}$ is to be replaced by $(6\pi^2/q)^{\frac23}$ (c.f.~\cite{Lieb1981b}) -- 

%----------------  proofs of the upper bounds in the introduction  -----------------------------------

\begin{proof}[\textbf{Proof of Theorem~\ref{main1}.}]
Let $g\in L^2(\RR^3)$ be given by $g(x)=(2\pi)^{-3/4}e^{-x^2/4}$ and
set $g_{\eps}(x)=\eps^{-3/2}g(x/\eps)$, so that $\|g_{\eps}\|=1$ for
all $\eps>0$. Let $\beta=\sqrt{2}\alpha-U\geq 0$. If $\beta=0$ then
$E_N\leq 0$ by Proposition~\ref{EN-bound}. Hence it remains to
consider the case $\beta>0$. Every density function $\rho_N\in L^1(\RR^3)$ with
$\|\rho_N\|_1=N$ is of the form $\rho_N(x)=N^2\rho(N^{1/3}x)$ with $\|\rho\|_1=1$.
From Proposition~\ref{EN-bound} and Lemma~\ref{lm:costate} combined it
follows that 
\begin{equation}\label{m1-1}
    \beta^{-2}E_N \leq \tfrac{3}{10}(6\pi^2)^{\frac23}\int_{\RR^3}\rho_N(x)^{\frac53}dx -
  \tfrac{1}{2}\int_{\RR^6}\frac{\rho_{N,\eps}(x)\rho_{N,\eps}(y)}{|x-y|}dxdy +N\|\nabla g_{\eps}\|^2,
\end{equation}
where $\rho_{N,\eps}=\rho_N*|g_{\eps}|^2$. Suppose $1<\mu<6/5$ and let $f(k):=\widehat{|g|^2}=e^{-k^2/2}$.
Then $\widehat{\rho_{N,\eps}}(k)=\widehat{\rho_N}(k)\widehat{|g_{\eps}|^2}(k)=\widehat{\rho_N}(k)f(\eps
k)$ and 
\begin{equation}\label{m1-2}
   \sup_{k\neq 0}\frac{1-|f(k)|^2}{|k|^{\mu-1}} \leq 1.
\end{equation}
By definition of $f$, by \eqref{m1-2}, and by Lemma~\ref{lm:rho}
\begin{eqnarray*}
\lefteqn{\int_{\RR^6}\frac{\rho_{N}(x)\rho_{N}(y)}{|x-y|}dxdy-\int_{\RR^6}\frac{\rho_{N,\eps}(x)\rho_{N,\eps}(y)}{|x-y|}dxdy}\\
   &=&\frac{1}{2\pi^2}\int_{\RR^3}(1-|f(\eps k)|^2)\frac{|\widehat{\rho_N}(k)|^2}{|k|^2}\,dk\\
   &=& \frac{1}{2\pi^2}\eps^{\mu-1}\int_{\RR^3}\frac{1-|f(\eps k)|^2}{|\eps
     k|^{\mu-1}}\frac{|\widehat{\rho_N}(k)|^2}{|k|^{3-\mu}}\,dk\\
   &\leq&
   \frac{1}{2\pi^2}\eps^{\mu-1}\int_{\RR^3}\frac{|\widehat{\rho_N}(k)|^2}{|k|^{3-\mu}}\,dk\\
   &=& N^{2+\frac{\mu}{3}}\eps^{\mu-1}2(2\pi)^{\mu-2}\frac{c_{\mu}}{c_{3-\mu}}\int_{\RR^6}\frac{\rho(x)\rho(y)}{|x-y|^{\mu}}\,dxdy.
\end{eqnarray*}
Combining this estimate with \eqref{m1-1}, we see that 
\begin{eqnarray*}
    \beta^{-2}E_N &\leq & N^{\frac73}\cE_\PTF(\rho) + N\eps^{-2}\|\nabla g\|^2 \\
   && + N^{2+\frac{\mu}{3}}\eps^{\mu-1}(2\pi)^{\mu-2}\frac{c_{\mu}}{c_{3-\mu}}\int_{\RR^6}\frac{\rho(x)\rho(y)}{|x-y|^{\mu}}\,dxdy
\end{eqnarray*}
for all $\rho\in L^1(\RR^3)$ with $\|\rho\|_1=1$. If $\{\rho_n\}\subset
L^1(\RR^n)$ is a minimizing sequence, $\cE_\PTF(\rho_n)\to E_\PTF$ as
$n\to\infty$, then  $\|\rho_n\|_{5/3}$ is uniformly bounded
by~\eqref{a1}, and hence so is the term $\int
\rho_n(x)\rho_n(y)/|x-y|^{\mu}dxdy$ for $\mu<6/5$. Therefore, in the limit $n\to\infty$,
we obtain
\begin{eqnarray*}
   \beta^{-2}E_N \leq N^{\frac73} E_\PTF + \tfrac{1}{4}N\eps^{-2} + N^{2+\frac{\mu}{3}}\eps^{\mu-1}C_{\mu}
\end{eqnarray*}
where the constant $C_{\mu}$ is finite for $\mu<6/5$ and $\|\nabla g\|^2=1/4$ was used. Upon optimizing with respect to $\eps$
we arrive at
$$
   \beta^{-2}E_N \leq N^{\frac73}E_\PTF +
   N^{\frac{9+5\mu}{3+3\mu}} D_{\mu}
$$
with a new constant $D_{\mu}$. This bound with the choice $\mu=37/31<6/5$
proves Theorem~\ref{main1}.
\end{proof}

\begin{proof}[\textbf{Proof of Theorem~\ref{main4}}]
We only need to prove that $E_1\leq -\alpha$. The bound $E_N\leq
-\alpha N$ will then follow from $E_{N+M}\leq E_N+E_M$ as pointed out
in the introduction.

Following Nelson \cite{Nelson1964} we introduce
$$
   B_{\Lambda}:= \frac{\sqrt{\alpha}}{c_0} \int_{|k|\leq
     \Lambda}\frac{1}{\i(1+\frac{k^2}2)|k|}\big[e^{\i k\cdot x}a(k) +
   e^{-\i k\cdot x}a^{*}(k)\big]\, dk.
$$
Then
\begin{equation}
  e^{\i B_{\Lambda}} H_{1,\Lambda} e^{-\i B_{\Lambda}}=\tfrac{1}{2}
  \left(p^2 + 2a^{*}\cdot p + 2p\cdot a+ a^2 + (a^{*})^2+2a^{*}a\right)
  +H_\ph -\alpha e_{\Lambda},\label{Hnelson}
\end{equation}
where
\begin{eqnarray*}
    a &:=& \frac{\sqrt{\alpha}}{c_0} 
\int_{|k|\leq \Lambda}\frac{k}{(1+\frac{k^2}2)|k|}e^{\i kx}a(k)\, dk.\\
     e_{\Lambda}  &:=&  \frac{1}{c_0^2} \int_{|k|\leq \Lambda}\frac{1}{|k|^2(1+\frac{k^2}2)}\, dk
\end{eqnarray*}
From \eqref{Hnelson} we see that, for all normalized $f\in L^2(\RR^3)$,
\begin{equation}\label{fxOmega}
     \la f\otimes \Omega, e^{\i B_{\Lambda}} H_{1,\Lambda}
     e^{-\i B_{\Lambda}} f\otimes\Omega\ra = \big\la f,
     (-\tfrac12 \Delta)f\big\ra - \alpha e_{\Lambda}
\end{equation}
where $\Omega\in\cF$ denotes the vacuum vector. 
Since $\inf\sigma(-\Delta)=0$ it follows from \eqref{fxOmega} that $E_{1,\Lambda}\leq -\alpha
e_{\Lambda}$, where
$$
  \lim_{\Lambda\to\infty} e_{\Lambda} 
  = \frac{1}{c_0} \int_{\RR^3}\frac{1}{|k|^2(1+\frac{k^2}2)}\, dk =1.
$$
This concludes the proof of the first bound in
Theorem~\ref{main4}.

A result similar to $E_{N+M}\leq E_N+E_M$ is expressed
by Theorem~6 in \cite{Griesemer2004}. A copy of the proof of that
theorem, with small modifications due to the differences of the
Hamiltonians, also proves the desired bound here. In fact, the main part of
the proof of \cite[Theorem~6]{Griesemer2004} is Equation~(19) and the equation thereafter, which show
that the interaction between electrons mediated by bosons decreases
with increasing particle separation. This part remains valid for
the coupling function $\chi_{|k|\leq\Lambda}/(c_0|k|)$ of
the Hamiltonian $H_{N,\Lambda}$. Other parts of the proof are simplified due
to the fact the phonon dispersion relation $\omega_{LO}$ is constant
and hence a local operator with respect to the boson position as
measured by $\i \nabla_k$.
\end{proof}

%--------------------------------------------------------------------------------------------------------------------------------------------------

%%%%%%%%%%%%%%%%%%%%%%%%%%%%%%%%%%%%%%%%%%%%%%%%%%%%%%%%
%%%%%%%%%%%%%%%%%%%%%%%%%%%%%%%%%%%%%%%%%%%%%%%%%%%%%%%%
\secct{Lower bounds on $E_N$}
%%%%%%%%%%%%%%%%%%%%%%%%%%%%%%%%%%%%%%%%%%%%%%%%%%%%%%%%
%%%%%%%%%%%%%%%%%%%%%%%%%%%%%%%%%%%%%%%%%%%%%%%%%%%%%%%%

In this section we prove Theorems~\ref{main1},
and \ref{main3}. The first step is to make sure that phonons with large momenta 
contribute to lower order in $N$. To this end,
for given $K,\Lambda,\delta,\kappa>0$, we define the operator
\begin{eqnarray*}
H_{N,\Lambda,K} & := &  -\tfrac12(1- \CI)\sum_{\ell=1}^N \Delta_\ell + 
(1-\delta)H_\ph+ U V_\C \\
&&  +
\sqrt{\norm}\sum_{\ell=1}^N\int_{|k|\leq\Lambda}\frac{e^{-\frac{|k|^2}{4K^2}}}{c_0|k|}
\big[e^{\i k\cdot x_\ell}a(k) + e^{-\i k\cdot x_\ell}a^*(k)\big]\,dk.
\end{eqnarray*} 
Of course, later on, $\delta,\kappa\in(0,1)$ and $K<\Lambda\to\infty$.
The following result, in the case $N=1$, is essentially due to Lieb and Thomas 
\cite{LiebThomas1997}.
While a sharp cutoff $|k|\leq K$ is used in  \cite{LiebThomas1997}, we work with 
a Gaussian cutoff since we need the Fourier transform of the cutoff to
be positive.

\begin{lemma}\label{StepI}
Suppose $K,\Lambda,\alpha$ and $U$ are positive, 
$0<\delta<1$ and let $\kappa :=
\frac{8\alpha N}{3K\delta}I_{\infty}$, where
$I_{\infty}:=(\sqrt{2}-1)/\sqrt{\pi}$. Then 
\begin{equation}
H_{N,\Lambda} \geq H_{N,\Lambda,K} - \frac{3}{2\delta}.
\end{equation}
\end{lemma}

\begin{proof}
For each $\ell\in\{1,\ldots, N\}$, we introduce three high
momenta modes by
\begin{eqnarray*}
    Z_j^{(\ell)} &:=&  \int_{\RR^3} \overline{T_{j}^{(\ell)}(k)}a(k)dk,\qquad
    j\in\{1,2,3\}, \\
    T_j^{(\ell)}(k) &:=& \sqrt{\alpha}\chi_{\Lambda}(k)\frac{1-e^{-\frac{|k|^2}{4K^2}}}{c_0|k|^3}k_je^{-\i
      k\cdot x_\ell},
\end{eqnarray*}
$k_j\in\RR$ being the $j$-th component of $k\in\RR^3$ and
$\chi_{\Lambda}$ the characteristic function of the set $\{|k|\leq \Lambda\}$. For later
use we compute the inner product of two functions
$T_{j}^{(\ell)}$. By straightforward computations,
\begin{equation}\label{T-norm}
   \int_{\RR^3} \overline{T_{j}^{(\ell)}(k)}
   T_{j'}^{(\ell)}(k)\, dk = \delta_{jj'} \frac{\alpha}{3K} I_{\frac{\Lambda}{K}},
\end{equation}
where
$$
   I_R:=\frac{\sqrt{2}}{\pi}\int_0^R \frac{(1-e^{-\frac{s^2}4})^2}{s^2}ds.
$$
Note that $4\pi/c_0^2=\sqrt{2}/\pi$ and that $I_{\infty}=\lim_{R\to\infty}I_R=(\sqrt{2}-1)/\sqrt{\pi}$
as defined in the statement of the lemma. By definition of
$H_{N,\Lambda,K}$,
\begin{eqnarray}\label{H-splitting}
     H_{N,\Lambda} &=& H_{N,\Lambda,K} + \sum_{\ell=1}^N
     \Big(-\frac{\kappa}{2}\Delta_\ell+ I_{K,\Lambda}^{(\ell)}\Big) + \delta H_\ph\\
 I_{K,\Lambda}^{(\ell)} &:=& \sqrt{\norm}\int_{|k|\leq\Lambda}\frac{1-e^{-\frac{|k|^2}{4K^2}}}{c_0|k|} \big[e^{\i
k\cdot x_\ell}a(k)+h.c.\big]dk\nonumber
\end{eqnarray}
where we introduced the operators $I_{K,\Lambda}^{(\ell)}$ associated
with the ultraviolet part of the electron-phonon interaction.
The key ingredient of this proof is that
\begin{equation}\label{key-id}
I_{K,\Lambda}^{(\ell)}=
\sum_{j=1}^3\big[p_{\ell,j},Z_j^{(\ell)}-Z_j^{(\ell)*}\big]
\end{equation}
where $p_{\ell,j}:=-i\partial/\partial_{x_{\ell,j}}$. This identity
implies that
\begin{eqnarray}
\big|\la\eta,I_{K,\Lambda}^{(\ell)}\eta\ra\big|  & \leq & 2\sum_{j=1}^3\|
p_{\ell,j}\eta\|
\|(Z_j^{(\ell)}-Z_j^{(\ell)*})\eta\|\nonumber\\
& \leq & \frac{\kappa}2 \la \eta, -\Delta_\ell\eta\ra +
\frac2{\kappa}\sum_{j=1}^3\la\eta,-(Z_j^{(\ell)}-Z_j^{(\ell)*})^2\eta\ra\nonumber\\
& \leq & \frac{\kappa}2 \la \eta, -\Delta_\ell\eta\ra +
\frac4{\kappa}\sum_{j=1}^3\la\eta, (Z_j^{(\ell)*}Z_j^{(\ell)} +
Z_j^{(\ell)} Z_j^{(\ell)*})\eta\ra,\label{key-est}
\end{eqnarray}
where $\kappa>0$ is to be selected, and the estimate
\[
|\la\eta, (Z_j^{(\ell)})^2\eta\ra| \leq
\|Z_j^{(\ell)*}\eta\|\|Z_j^{(\ell)}\eta\| \leq \tfrac12
\la\eta,(Z_j^{(\ell)} Z_j^{(\ell)*} + Z_j^{(\ell)*}
Z_j^{(\ell)})\eta\ra
\]
was used. From \eqref{T-norm} and $I_{\Lambda/K}\leq I_{\infty}$ it is
clear that
\begin{eqnarray}
\sum_{j=1}^3( Z_j^{(\ell)*}Z_j^{(\ell)} + Z_j^{(\ell)}Z_j^{(\ell)*})
& = & \sum_{j=1}^3 2Z_j^{(\ell)*}Z_j^{(\ell)} + \big[Z_j^{(\ell)},Z_j^{(\ell)*}\big]\nonumber\\
& \leq & \frac{2\alpha I_\infty}{3K}H_{\ph} +
\frac{\alpha}{K}I_{\infty}.\label{n-order}
\end{eqnarray}
Combining \eqref{key-est} and \eqref{n-order} we arrive at
$$
\pm\sum_{\ell=1}^N I_{K,\Lambda}^{(\ell)} \leq  \frac{\kappa}{2}\sum_{\ell=1}^N(-\Delta_\ell) + 
\frac{8\alpha NI_{\infty}}{3\kappa K} H_\ph + \frac{4\alpha NI_{\infty}}{\kappa K},
$$
which, by \eqref{H-splitting} and the choice $\kappa=8\alpha N I_{\infty}/(3K\delta)$, proves the lemma.
\end{proof}

%--------------  new lemma -------------------------------------------------

\begin{lemma}\label{StepII}
Suppose $K,\Lambda,\alpha,U$ and $\kappa$ are positive,
and $0<\delta\leq 1/2$. Then 
\begin{equation}
  H_{N,\Lambda,K} \geq
  -\tfrac12(1- \CI)\sum_{\ell=1}^N \Delta_\ell -\left(\frac{\sqrt{2}\alpha}{1-\delta}-U\right)V_\C - \frac{2\alpha NK}{\sqrt{\pi}}.
\end{equation}
\end{lemma}

\begin{proof}
By completing the square in annihilation and creation operators, that
is, by using Lemma~\ref{squares}, we see that
\begin{eqnarray}
\lefteqn{(1-\delta)H_\ph
+\frac{\sqrt{\alpha}}{c_0}\sum_{\ell=1}^N\int_{|k|\leq \Lambda}
\frac{e^{-\frac{|k|^2}{4K^2}}}{|k|}
\big[e^{\i k\cdot x_\ell}a(k) + e^{-\i k\cdot x_\ell}a^*(k)\big]\,dk. }\nonumber\\
  &\geq & -\frac{\alpha}{(1-\delta)c_0^2}\sum_{j,\ell=1}^N
  \int_{\RR^3}\frac{e^{-\frac{|k|^2}{2K^2}}}{|k|^2}  e^{\i
    k\cdot(x_j-x_\ell)}\, dk\nonumber\\
  &=& - \frac{2\alpha}{(1-\delta)c_0^2}\sum_{j<\ell}
  \int_{\RR^3}\frac{e^{-\frac{|k|^2}{2K^2}}}{|k|^2}  e^{\i
    k\cdot(x_j-x_\ell)}\, dk - \frac{\alpha NK}{(1-\delta)\sqrt{\pi}}.\label{stepIIa}
\end{eqnarray}
The integral in \eqref{stepIIa} represents the electrostatic
energy of two spherically symmetric, non-negative charge distributions centered at
$x_j$ and $x_\ell$, respectively, each distribution having total charge
one, see \eqref{Plancherel}. 
Hence Newton's theorem, \cite[Theorem 9.7]{LiebLoss2001}, implies that 
$$
     \int_{\RR^3}\frac{e^{-\frac{|k|^2}{2K^2}}}{|k|^2}  e^{\i
    k\cdot(x_j-x_\ell)}\, dk  \leq \frac{2\pi^2}{|x_j-x_\ell|}.
$$
Since $c_0^2 = 2\pi^2\sqrt{2}$, it follows that \eqref{stepIIa} is
bounded below by
$$
       -\frac{\sqrt{2}\alpha}{1-\delta}V_\C - \frac{\alpha NK}{(1-\delta)\sqrt{\pi}},
$$
which proves the lemma.
\end{proof}

%------------------  proofs of the theorems on the lower bounds -------------------------------

\begin{proof}[\textbf{Proof of Theorem~\ref{main2}.}]
We shall combine the Lemmas~\ref{StepI} and \ref{StepII} with suitable
choices for $\delta$ and $K$. First, suppose that $0<\delta\leq 1/2$
and that $\kappa\in(0,1)$. Since $\sqrt{2}\alpha-U\geq 0$, 
by assumption of Theorem~\ref{main2}, 
the constant multiplying the potential $V_\C$ in Lemma~\ref{StepII} is
positive, and hence, after the scaling transformation
$$
    x\to \left(\frac{\sqrt{2}\alpha}{1-\delta}-U\right)^{-1}(1-\kappa)x
$$
we may apply \eqref{Gravitation} and find that 
\begin{eqnarray*}
    H_{N,\Lambda,K} &\geq &
    -\frac{\left(\frac{\sqrt{2}\alpha}{1-\delta}-U\right)^2}{1-\kappa}
    \inf\sigma\Big(\sum_{\ell=1}^N -\frac12\Delta_\ell-V_\C\Big) 
    - \frac{2\alpha NK}{\sqrt{\pi}}\\
     &\geq & -C_\G\frac{\left(\frac{\sqrt{2}\alpha}{1-\delta}-U\right)^2}{1-\kappa} N^{\frac73} - \frac{2\alpha NK}{\sqrt{\pi}},
\end{eqnarray*}
where $C_\G$ is chosen such that \eqref{Gravitation} holds true.

We now make the choices 
$$
 \delta = \tfrac12 N^{-\frac19} \ \ \mathrm{and} \ \ K = \tfrac13 32
   I_\infty\alpha N^{1+\frac29}, 
$$
which imply that $\kappa$, as defined in Lemma~\ref{StepI}, obeys 
$\kappa= \tfrac{1}{2} N^{-1/9}=\delta$.
Using that $(1-t)^{-1} \leq 1+2t$, for $0\leq t\leq 1/2$, that
$U\leq \sqrt{2}\alpha$,  and  $I_{\infty}/\sqrt{\pi}=(\sqrt{2}-1)/\pi\leq 1/(2\pi)$, we find that
\begin{eqnarray*}
     H_{N,\Lambda,K}& \geq & -C_\G\big[\sqrt{2}\alpha(1+2\delta)-U\big]^2 
        (1+2\kappa) N^{\frac73} - \frac{32}{3\pi}\alpha^2N^{2+\frac29}\\
     &\geq &   -C_\G\big[(\sqrt{2}\alpha-U)^2+16\alpha^2\delta\big]
     (1+2\kappa) N^{\frac73}- \frac{32}{3\pi}\alpha^2 N^{\frac73-\frac19}\\
      &\geq &
      -C_\G\big[(\sqrt{2}\alpha-U)^2N^{\frac73}+18\alpha^2N^{\frac73-\frac19}\big] 
      - \frac{32}{3\pi}\alpha^2 N^{\frac73-\frac19} \\
     &=& -C_{\G}(\sqrt{2}\alpha-U)^2N^{\frac73} -
     \Big(18C_\G+\frac{32}{3\pi}\Big)\alpha^2 N^{\frac73-\frac19}.
\end{eqnarray*}
\end{proof}

\begin{proof}[\textbf{Proof of Theorem~\ref{main3}.}] Finally we
  consider the case, where $U-\sqrt{2}\alpha>0$. In
  Lemma~\ref{StepII} we choose $\delta = (U-\sqrt{2}\alpha)/(2U)$ and
$K = 8\alpha NI_\infty/(3\delta)$, so that $\kappa=1$ in
Lemma~\ref{StepI}, and 
$$
 U - \frac{\sqrt{2}\alpha}{1-\delta} = \frac{U-\sqrt{2}\alpha}{2(1-\delta)}>0.
$$
From Lemma~\ref{StepI} and  Lemma~\ref{StepII} it hence follows that
\begin{equation*}
   E_N \geq -\frac{2\alpha NK}{\sqrt{\pi}} -\frac{3}{2\delta}
   = -\frac{16I_{\infty}}{3\sqrt{\pi}}\alpha^2 N^2\frac{2U}{U-\sqrt{2}\alpha} -\frac{3U}{U-\sqrt{2}\alpha},
\end{equation*}
where $I_{\infty}/\sqrt{\pi}=(\sqrt{2}-1)/\pi\leq 1/(2\pi)$.
\end{proof}

%%%%%%%%%%%%%%%%%%%%%%%%%%%%%%%%%%%%%%%%%%%%%%%%%%%%%%%%
%%%%%%%%%%%%%%%%%%%%%%%%%%%%%%%%%%%%%%%%%%%%%%%%%%%%%%%%
\begin{appendix}
\secct{Auxiliary Results} \label{app-Pekar}
%%%%%%%%%%%%%%%%%%%%%%%%%%%%%%%%%%%%%%%%%%%%%%%%%%%%%%%%
%%%%%%%%%%%%%%%%%%%%%%%%%%%%%%%%%%%%%%%%%%%%%%%%%%%%%%%%

\begin{lemma}\label{lm:rho}
Suppose that $\rho\in L^1(\RR^3)\cap L^{5/3}(\RR^3)$, $\rho\geq 0$,
$0<\mu<6/5$, and let $\rho_{N}(x)=N^{2}\rho(N^{1/3}x)$. Then
\begin{eqnarray}
 \int_{\RR^6} \frac{\rho(x)\rho(y)}{|x-y|^{\mu}}\, dxdy &\leq & 
   a_{\mu}\|\rho\|_1^{2-\frac{5\mu}{6}}\|\rho\|_{\frac53}^{\frac{5\mu}{6}}\label{a1}\\
 \int_{\RR^6} \frac{\rho(x)\rho(y)}{|x-y|^{\mu}}\, dxdy 
&=& (2\pi)^{-\mu}\frac{c_{3-\mu}}{c_{\mu}}\int_{\RR^3}\frac{|\hat{f}(k)|^2}{|k|^{3-\mu}}dk\label{a2}\\
  \int_{\RR^6} \frac{\rho(x)\rho(y)}{|x-y|}\, dxdy 
&=& \frac{1}{2\pi^2}\int_{\RR^3}\frac{|\hat{f}(k)|^2}{|k|^2}dk\label{Plancherel}\\
  \int_{\RR^6} \frac{\rho_N(x)\rho_N(y)}{|x-y|^{\mu}}\, dxdy 
&=& N^{2+\frac{\mu}{3}}\int_{\RR^6} \frac{\rho(x)\rho(y)}{|x-y|^{\mu}}\, dxdy \label{a4}
\end{eqnarray}
where
$$
  a_\mu :=\left(\frac{4\pi}{3}\right)^{\frac{\mu}{3}}\left(\frac{6}{5\mu}\right)^{1+\frac{\mu}{3}}
 \left(\frac{6}{5\mu}-1\right)^{-1+\frac{\mu}{2}},\qquad c_{\mu}:=\pi^{-\frac{\mu}{2}}\Gamma(\frac{\mu}{2})
$$
in \eqref{a1} and \eqref{a2}, respectively.
\end{lemma}

Inequality \eqref{a1}, in the special case $\mu=1$, implies that 
$\cE_\PTF(\rho)$ is bounded below, and moreover, that
$\|\rho\|_{5/3}$ is bounded unformly for densities $\rho$ with  $\|\rho\|_1=1$ 
and $\cE_\PTF(\rho)\leq E_\PTF+1$.

\begin{proof}[Proof of Lemma~\ref{lm:rho}]
For each $R>0$, by H\"older's inequality,
\begin{align*}
\int_{\RR^3}\frac{\rho(y)}{|x-y|^{\mu}}dy &=\int_{|x-y|\leq R}\frac{\rho(y)}{|x-y|^{\mu}}dy +
\int_{|x-y|\geq R}\frac{\rho(y)}{|x-y|^{\mu}}dy\\
 &\leq \left(\frac{8\pi}{6-5\mu}\right)^{\frac25}R^{\frac65-\mu}\|\rho\|_{\frac53} +R^{-\mu}\|\rho\|_1.
\end{align*}
By optimizing this bound w.r.to $R>0$ we obtain \eqref{a1}.
Equation~\eqref{a2} follows from
\cite[Corollary~5.10]{LiebLoss2001}. The factor $(2\pi)^{-\mu}$ stems from
the differences in the definition of the
Fouriertransform. Equation~\eqref{Plancherel} is the important special case
$\mu=1$ from \eqref{a2}, and \eqref{a4} is straightforward to
verify by a change of variables. 
\end{proof}

\begin{lemma}\label{squares}
Suppose $f\in L^2(\RR^3)$. Then, for every $\delta>0$,
$$
    \int_{\RR^3}\big[\delta a^{*}(k)a(k)+\overline{f(k)}a(k) +
    f(k)a^{*}(k)\big]\,dk\ \geq\ -\frac{1}{\delta}\|f\|^2
$$
and the lower bound is attained by the expectation value in the
coherent state $\eta\in\cF$, $\|\eta\|=1$, defined by $a(k)\eta=-\delta^{-1}f(k)\eta$.
\end{lemma}

\begin{proof}
By completeing the square in creation and annihilation operators
\begin{eqnarray*}
   \lefteqn{\int_{\RR^3}\big[\delta a^{*}(k)a(k)+\overline{f(k)}a(k) +
    f(k)a^{*}(k)\big]\,dk} \\
    &=& \int_{\RR^3}\Big[
    \delta\Big(a^{*}(k)+\frac{\overline{f(k)}}{\delta}\Big)
    \Big(a(k)+\frac{f(k)}{\delta}\Big) -\frac{|f(k)|^2}{\delta}\Big]\,dk\\
    &\geq & -\frac{1}{\delta}\|f\|^2
\end{eqnarray*}
\end{proof}

\begin{proposition}\label{mainA}
Suppose that $\sqrt{2}\alpha\leq U$. Then for all $N,\Lambda>0$, and
all $f\in \wedge^N L^2(\RR^3),\eta\in \cF$ with $\|f\|=\|\eta\|=1$,
$$
     \langle f\otimes\eta, H_{N,\Lambda}f\otimes\eta\rangle \geq
     -c_\romL \frac{40}{3}\left(\frac{2}{3\pi}\right)^{\frac23}N,
$$
where $c_\romL=1.68$ or any other
constant for which the Lieb-Oxford inequality holds.
\end{proposition}

\begin{proof}
As in the proof of Theorem~\ref{main1}
\begin{eqnarray*}
   \lefteqn{\langle f\otimes\eta, H_{N,\Lambda}f\otimes\eta\rangle} \\
   &=& \langle f, (-\tfrac{1}{2}\Delta +UV_\C)f\rangle
   +\langle\eta,(H_\ph+\sqrt{\alpha}\phi_{\Lambda}(\rho))\eta\rangle\\
   &\geq & \langle f, (-\tfrac{1}{2}\Delta +UV_\C)f\rangle -\sqrt{2}\alpha\tfrac{1}{2}\int_{\RR^6}\frac{\rho(x)\rho(y)}{|x-y|}dxdy.
\end{eqnarray*}
Using the Lieb-Thirring \cite[Theorem~2.15]{LiebLoss2001} and the Lieb-Oxford inequalities \cite{Lieb1981b} we find that
\begin{eqnarray}\nonumber
   \lefteqn{\langle f\otimes\eta, H_{N,\Lambda}(f\otimes\eta)\rangle} \\
   &\geq & c_{\LT}\int_{\RR^3}\rho(x)^{\frac53}dx +
   (U-\sqrt{2}\alpha)\tfrac{1}{2}\int_{\RR^6}\frac{\rho(x)\rho(y)}{|x-y|}dxdy-c_\romL\int_{\RR^3}\rho(x)^{\frac43}dx\label{stable-TF}
\end{eqnarray}
where $c_{\LT}=\frac{3}{10}(\frac{3\pi}{2})^{2/3}$ and $c_\romL=1.68$ or any other constant for which the Lieb-Oxford
inequality is satisfied. From
$\rho(x)^{4/3}=\rho(x)^{5/6}\rho(x)^{1/2}$ and the Cauchy-Schwarz inequality, for every
$\eps>0$,
\begin{align}
    \int_{\RR^3}\rho(x)^{\frac43}dx &\leq \left(
      \int_{\RR^3}\rho(x)^{\frac53}dx\right)^{\frac12}\left(
      \int_{\RR^3}\rho(x)dx\right)^{\frac12}\nonumber\\
    &\leq \tfrac{1}{2}\left(\eps \int_{\RR^3}\rho(x)^{\frac53}dx
      +\frac{1}{\eps}\int_{\RR^3}\rho(x)dx \right). \label{four-five}
\end{align} 
The estimates \eqref{stable-TF} and \eqref{four-five} with
$\eps=2c_{\LT}/Uc_\romL$ prove the proposition.
\end{proof}

\end{appendix}

%%%%%%%%%%%%%%%%%%%%%%%%%%%%%%%%%%%%%%%%%%%%
%\bibliography{reflist}
%\end{document}
%%%%%%%%%%%%%%%%%%%%%%%%%%%%%%%%%%%%%%%%%%%%

\providecommand{\bysame}{\leavevmode\hbox to3em{\hrulefill}\thinspace}
\providecommand{\MR}{\relax\ifhmode\unskip\space\fi MR }
\providecommand{\MRhref}[2]{%
  \href{http://www.ams.org/mathscinet-getitem?mr=#1}{#2}
}
\providecommand{\href}[2]{#2}

\end{document}